\documentclass[11pt,a4paper]{article}
\usepackage[T1]{fontenc}
\usepackage{amsmath,amssymb,amsthm,pxfonts}
\usepackage[margin=1in]{geometry}
\usepackage{microtype}
\usepackage[driverfallback=dvipdfmx,colorlinks=true,linkcolor=blue,citecolor=blue,urlcolor=blue]{hyperref}
\hypersetup{pdftitle={Fundamental Theorem of Discrete Asset Pricing},
  pdfauthor={Masaaki Fukasawa},
  pdfsubject={Continuous-time multidimensional asset pricing on a locally finite state space}}
\newtheorem{theorem}{Theorem}
\newtheorem{lemma}[theorem]{Lemma}
\theoremstyle{remark}
\newtheorem*{remark}{Remark}
\newcommand{\Pp}{\mathbb P}
\newcommand{\Qq}{\mathbb Q}
\newcommand{\E}{\mathbb E}
\newcommand{\F}{\mathcal F}
\newcommand{\one}{\mathbf 1}
\newcommand{\var}{\operatorname{Var}}
\newcommand{\conv}{\operatorname{conv}}
\newcommand{\as}{\text{a.s.}}

\title{Fundamental Theorem of Discrete Asset Pricing}
\author{Masaaki Fukasawa\thanks{The University of Osaka,
Email: fukasawa@sigmath.es.osaka-u.ac.jp
}}
\date{}

\begin{document}
\maketitle

\begin{abstract}
Motivated by discrete price quotes in continuous time, we seek a fundamental
theorem based on no arbitrage rather than a condition on limiting trading
opportunities. We consider a vector of discounted prices taking values in a locally finite
subset of $\mathbb R^d$ in continuous time. For a locally bounded c\`adl\`ag price, we prove that no arbitrage, no free
lunch with vanishing risk, no free lunch, and the existence of an equivalent
local martingale measure are equivalent. An elementary finite-dimensional
reduction of portfolios yields a common pathwise bound on admissible wealth.
This bound reduces closedness to convex averaging in a
weighted Hilbert space, after which the Kreps--Yan theorem applies. Trading
gains are finite Stieltjes sums, and the proof uses no It\^o integration.
\end{abstract}

\section{Introduction}

The fundamental theorem of asset pricing (FTAP) relates the absence of
arbitrage to the existence of an equivalent martingale measure. The
martingale approach was developed by Harrison and Kreps \cite{HK} and
Harrison and Pliska \cite{HP}. For finitely many deterministic trading dates,
Dalang, Morton and Willinger \cite{DMW} proved the equivalence between
no arbitrage (NA) and an equivalent martingale measure for vector prices.
Hilbert-space and utility-based proofs were given by Schachermayer \cite{Sch}
and Rogers \cite{Rogers}; see also Jacod and Shiryaev \cite{JS}.

NA expresses a simple economic principle: a trader should not be able to
start with zero capital and end with a nonnegative gain that is positive
with positive probability. Admissibility requires a fixed lower bound on
gains throughout the horizon, excluding unlimited borrowing through doubling
strategies. In general continuous-time models, however, NA alone is
insufficient for an FTAP. Stronger conditions also rule out opportunities
that appear as limits of trading strategies.

Kreps \cite{Kreps} introduced no free lunch (NFL), which uses the weak-star
closure of attainable dominated claims. No free lunch with vanishing risk
(NFLVR) instead uses norm closure. Delbaen and Schachermayer \cite{DS}
proved the FTAP with NFLVR for locally bounded semimartingales; weak-star
closedness is a central step, also discussed by Kabanov \cite{K}.
NFLVR is economically more transparent than NFL: the possible loss tends
to zero uniformly as the payoffs approach a nonzero nonnegative limit.
Nevertheless, it still concerns a limiting sequence rather than one
executable arbitrage. Can a natural assumption on prices recover an FTAP
based on ordinary NA while retaining continuous time?

Our motivation is that price quotes and bank-account interest change in
discrete increments, while trading and information arrival take place in
continuous time. We therefore assume that discounted prices take values in
a locally finite set: each bounded region contains only finitely many price
vectors. For locally bounded c\`adl\`ag prices with this property, we prove
that NA, NFLVR, NFL, and the existence of an equivalent local martingale
measure are equivalent. Neither the trading times nor the number of price
changes is fixed in advance.

As in Schachermayer \cite{Sch}, we use orthogonal projection and
Hilbert-space compactness. His setting has finitely many deterministic
trading dates. Here jumps occur at stopping times, their number has no
deterministic bound, and information may arrive between jumps. Holdings
at a jump can therefore use information unavailable at the preceding jump;
sampling prices at jump times does not reduce our strategy class to a
finite-period model. A common pathwise wealth bound handles this difficulty.

Cuchiero, Klein and Teichmann \cite{CKT} study discrete or imperfect price
observations through two filtrations, using an $L^p$ asymptotic no-free-lunch
condition and obtaining martingale optional projections onto the trading
filtration. We instead restrict traded price values, use one filtration and
lower-bound admissibility, and obtain a local martingale measure for the
price process itself from ordinary NA. The resulting gains are finite
Stieltjes sums, requiring no It\^o integration. The main idea is to remove
redundant portfolio positions by finite-dimensional geometry, then use NA
to obtain a bound common to all admissible wealth processes. This reduces
the closedness argument to convex averaging in a weighted Hilbert space.
The geometry is isolated in a separate lemma, and the averaging and
separation arguments are supplied in the appendix.

\section{Model and main theorem}

Fix $T<\infty$ and a filtered probability space
$(\Omega,\F,(\F_t)_{0\le t\le T},\Pp)$ satisfying the usual conditions, with
$\F=\F_T$. The bank account is the num\'eraire and has discounted value one.
Fix $d<\infty$ and a deterministic set $E\subset\mathbb R^d$ such that
$E\cap\{x:|x|\le b\}$ is finite for every $b<\infty$, where $|\cdot|$ is
the Euclidean norm. Let $S$ be an adapted c\`adl\`ag discounted price vector with deterministic
$S_0$, taking values in $E$ outside one null set. Assume local
boundedness: there are stopping times $\sigma_m\uparrow T$, eventually equal
to $T$ almost surely, for which each stopped process $S^{\sigma_m}$ is bounded
by a deterministic constant $b_m$.

Every c\`adl\`ag path is bounded on $[0,T]$, so its values lie in a finite
subset of $E$. Distinct values on that path have a positive minimum
separation. Jump times cannot accumulate, by left limits and right
continuity, and the path is constant between jumps. Every predictable $\mathbb R^d$-valued
$H$ consequently has the Stieltjes gain
\begin{equation}\label{gain}
 X_t=(H\cdot S)_t=\sum_{0<u\le t}\langle H_u,\Delta S_u\rangle.
\end{equation}
The successive price-jump times can be enumerated as stopping times.
Predictable holdings evaluated at these times are measurable at the
corresponding stopping times, so the finite sum \eqref{gain} defines an
adapted c\`adl\`ag gain process.
A gain is $a$-admissible if $X_t\ge-a$ for all $t$, outside one null set,
where $a>0$, and admissible if it is $a$-admissible for some $a>0$. Define
\[
 C=\{f\in L^\infty:f\le X_T\text{ for some admissible gain }X\}.
\]
The set $C$ is a convex cone containing $-L^\infty_+$: gains add and
scale with the strategies, their lower bounds add and scale accordingly,
and the zero strategy dominates every nonpositive claim.
All random variables are identified up to almost-sure equality. Define
\[
 \begin{gathered}
 \mathrm{NA}:\quad C\cap L^\infty_+=\{0\},\qquad
 \mathrm{NFLVR}:\quad \overline C^{\|\cdot\|_\infty}\cap L^\infty_+=\{0\},\qquad
 \mathrm{NFL}:\quad \overline C^{\sigma(L^\infty,L^1)}\cap L^\infty_+=\{0\}.
 \end{gathered}
\]
NFL uses the full weak-star closure, so NFL $\Rightarrow$ NFLVR $\Rightarrow$
NA. By truncation from above, NA is equivalent to the absence of a nonzero
nonnegative admissible terminal gain.

\begin{theorem}\label{main}
The following conditions are equivalent:
\[
 \mathrm{NA},\quad \mathrm{NFLVR},\quad \mathrm{NFL},\quad
 \exists\,\Qq\sim\Pp\text{ under which }S\text{ is a local martingale}.
\]
Under NA, $C$ is $\sigma(L^\infty,L^1)$-closed. If $S$ is bounded, the
martingale measure can be taken to make $S$ a true martingale.
\end{theorem}

\section{A common pathwise bound}

We first record a consequence of NA used throughout the proof. If an
admissible gain $X=H\cdot S$ satisfies $X_T\ge-a$, then $X\ge-a$ on the whole
interval. Indeed, suppose $X\ge-b$ and $A=\{X_s<-a\}\in\F_s$ has positive
probability for some deterministic $s<T$. Continuing the strategy only on
$A$ after $s$ gives the gain $\one_A(X_t-X_s)$ for $t\ge s$. It is bounded
below by $-b$ and has a strictly positive terminal value on $A$, and zero on
$A^c$, contradicting NA. Rational times and right continuity prove the claim
simultaneously at all times; this is the standard continuation argument
\cite[Lemma 2.1]{K}.

\subsection{The deterministic geometry}

Vector portfolios can contain large redundant positions: for $S=(Y,Y)$,
$(n,-n)$ has zero gain. The following deterministic lemma bounds a
replacement portfolio with the same relevant inner products.

\begin{lemma}[Finite-dimensional portfolio geometry]\label{geometry}
Let $J\subset\mathbb R^d\setminus\{0\}$ be finite, possibly empty.
There exist a vector $u_J \in \mathbb R^d$ and a constant
$\kappa_J\ge1$, depending only on $J$, such that
\[
 |u_J|\le1,\qquad \langle u_J,z\rangle\ge0\quad(z\in J)
\]
and 
\begin{equation}\label{geometry-bound}
 |P_Jh|\le\kappa_J w
 \end{equation}
 for every $w\ge0$ and $h\in\mathbb R^d$ satisfying
$\langle h,z\rangle\ge-w$ for all $z\in J^0$, 
where $J^0=\{z\in J:\langle u_J,z\rangle=0\}$ and
 $P_J$ is the orthogonal projection to $\operatorname{span}J^0$.
\end{lemma}

\begin{proof}
\emph{Step 1: choose a direction exposing as many vectors as possible.}
Among the vectors $u$ with $|u|\le1$ and $\langle u,z\rangle\ge0$ on $J$,
choose $u_J$ maximizing the number of strictly positive inner products.
The feasible set contains zero, and the possible counts belong to the finite
set $\{0,\ldots,\#J\}$, so a maximizer exists. Let
$L_J=\operatorname{span}J^0$ and let $P_J$ project orthogonally onto $L_J$.

\emph{Step 2: the remaining vectors have no one-sided direction in their span.}
Suppose $v\in L_J\setminus\{0\}$ and
$\langle v,z\rangle\ge0$ for every $z\in J^0$.
At least one of these inner products is positive: otherwise $v$ would be
orthogonal to $L_J$, and hence to itself. For sufficiently small
$\varepsilon>0$, the vector $u_J+\varepsilon v$ preserves all the positive
inner products on $J\setminus J^0$ and creates a positive inner product on
$J^0$. Here a single $\varepsilon$ works because $J$ is finite.
Dividing by $\max(1,|u_J+\varepsilon v|)$ preserves signs and restores the
norm bound. This contradicts Step 1.

\emph{Step 3: compactness gives a uniform quantitative bound.}
If $L_J=\{0\}$, take $\kappa_J=1$; then $P_J=0$ and the conclusion is
immediate. Otherwise the unit sphere of $L_J$ is compact, and Step 2 gives
\[
 \beta_J:=\min_{v\in L_J,\ |v|=1}
          \max_{z\in J^0}(-\langle v,z\rangle)>0.
\]
Indeed, the maximum is a continuous function of $v$, and it is strictly
positive at every point of that sphere. If $p=P_Jh\ne0$, apply this bound
to $v=p/|p|$. For some $z\in J^0$,
\[
 \beta_J|p|\le-\langle p,z\rangle
                  =-\langle h,z\rangle\le w.
\]
Then \eqref{geometry-bound} holds with $\kappa_J:=\max(1,\beta_J^{-1})$. The case $p=0$ is immediate.
\end{proof}

For example, if $J=\{(1,0),(-1,0),(0,1)\}$, choose $u_J=(0,1)$.
Then $P_J(h_1,h_2)=(h_1,0)$, and the constraints give $|h_1|\le w$,
although $h_2$ is unrestricted. NA will eliminate the exposed direction
from actual increments, leaving the zero face where projection applies.

\subsection{From geometry to a common wealth bound}

\begin{lemma}\label{bound}
Assume NA. There is a finite predictable process $c\ge1$, independent of
$a$ and $H$, with the following property. Set
\begin{equation}\label{R}
 R_t=\prod_{0<u\le t}(1+c_u|\Delta S_u|),\qquad
 V_t=\var(S)_t=\sum_{0<u\le t}|\Delta S_u|.
\end{equation}
For every $a>0$ and predictable $H$ with $W=a+H\cdot S\ge0$, there is a
predictable $G$ such that $G\cdot S=H\cdot S$ and
$|G_t|\le c_tW_{t-}$ for $0<t\le T$. Moreover,
\begin{equation}\label{bounds}
 0\le W_t\le aR_t,\qquad
 \var(H\cdot S)_t\le a(R_t-1),\qquad 0\le t\le T.
\end{equation}
All process identities and inequalities hold simultaneously outside a
null set, which may depend on $H$. The processes $R,V$ are finite and
$V\le R-1$.
\end{lemma}

\begin{proof}
Use the geometric objects of Lemma~\ref{geometry}. Only countably many
finite sets occur, so no random measurable selection theorem is needed.

\emph{Step 1: list the possible increments locally.}

Put $\sigma_0=0$, $E_m=E\cap\{x:|x|\le b_m\}$ and
$F_m=(E_m-E_m)\setminus\{0\}$. On the predictable interval
$(\sigma_{m-1},\sigma_m]$, every nonzero price increment belongs to the finite set $F_m$.
Indeed, both $S_t$ and $S_{t-}$ lie in $E_m$: the stopped bound controls
their norms, and the finite set $E_m$ is closed, so it contains the left limit.
For every subset $J\subseteq F_m$, fix the objects of Lemma~\ref{geometry}
once and for all. There are finitely many choices for each $m$, hence only
countably many choices altogether. Define on this interval
\[
 c_t=\max_{J\subseteq F_m}\kappa_J,\qquad
 J_t=\{z\in F_m:\langle H_t,z\rangle\ge-W_{t-}\},\qquad G_t=P_{J_t}H_t.
\]
\emph{Step 2: select the relevant geometric problem predictably.}
Use $c=1$ and $G=0$ on any remaining points. The intervals are predictable
because $\one_{[0,\sigma]}$ is left-continuous and adapted for every stopping
time $\sigma$. The pre-jump wealth $W_{t-}$ is predictable: it is left-continuous and
adapted, with initial value $a$. Each test
$\langle H_t,z\rangle\ge-W_{t-}$ is therefore predictable. Since $J_t$ has
finitely many possible values on each interval, selecting the previously
fixed $u_J$ and $P_J$ preserves predictability. This proves predictability
of $c$, $G$, and the auxiliary strategy $U_t=u_{J_t}$. At a price jump, nonnegative wealth gives
\[
 0\le W_t=W_{t-}+\langle H_t,\Delta S_t\rangle,
 \qquad\text{so}\qquad \Delta S_t\in J_t.
\]

\emph{Step 3: NA eliminates all exposed directions.}
Set $U=0$ elsewhere. Since $\langle U_t,\Delta S_t\rangle\ge0$ at every
price jump, its gain is nondecreasing and admissible. NA forces its terminal
gain to vanish. Every term in this finite sum is nonnegative, so each is
zero: actual increments lie in $J_t^0$, simultaneously almost surely.
By \eqref{geometry-bound},
\[
 \langle G_t,\Delta S_t\rangle=\langle H_t,\Delta S_t\rangle,
 \qquad |G_t|\le\kappa_{J_t}W_{t-}\le c_tW_{t-}.
\]
The norm bound holds at all covered times; the gain identity is needed
only at actual jumps. Summing the latter gives $G\cdot S=H\cdot S$.

\emph{Step 4: multiply the one-jump wealth bounds.}

At each jump, Cauchy--Schwarz and the replacement bound give
\[
 |\Delta(H\cdot S)_t|=|\langle G_t,\Delta S_t\rangle|
       \le c_tW_{t-}|\Delta S_t|,
 \qquad
 W_t\le W_{t-}(1+c_t|\Delta S_t|).
\]
Starting from $W_0=a$, multiply these inequalities in chronological order.
Induction gives $W\le aR$; between jumps both quantities are constant.
Also $|\Delta(H\cdot S)_t|\le aR_{t-}c_t|\Delta S_t|=a\Delta R_t$.
Summing proves \eqref{bounds}. Finitely many jumps give finiteness, and
$c\ge1$ gives $V_t\le\sum_{0<u\le t}c_u|\Delta S_u|\le R_t-1$.
\end{proof}

\section{Proof of the main theorem}

\begin{proof}[Proof of Theorem~\ref{main}]
\emph{(i) NA implies NFLVR.}
Assume NA. If $f_n\in C$
converges in $L^\infty$ to $f\ge0$, choose $\varepsilon_n\downarrow0$ with
$f_n\ge-\varepsilon_n$ and admissible gains $X^n_T\ge f_n$. The continuation
argument makes $X^n$ $\varepsilon_n$-admissible. Thus $f=0$, since
\[
 f_n\le X^n_T\le\varepsilon_n(R_T-1)\longrightarrow0\quad\as
\]

\smallskip
\noindent\emph{(ii) NA implies NFL.}

On the predictable sigma-field of $\Omega\times(0,T]$, define
\begin{equation}\label{mu}
 \mu(A)=\E\left[R_T^{-3}\int_0^T\one_Ac^{-2}\,dV\right].
\end{equation}
This is a finite measure because $c\ge1$ and $V_T\le R_T-1$. The terminal weight need
not be predictable: \eqref{mu} defines a measure on predictable sets, rather
than a predictable density process. For every $1$-admissible integrand $H$,
choose $G$ from Lemma~\ref{bound}. Then
\begin{equation}\label{L2}
 \|G\|_{L^2(\mu;\mathbb R^d)}^2
 =\E\left[R_T^{-3}\int_0^T |G|^2c^{-2}\,dV\right]
 \le\E[V_T/R_T]\le1,
\end{equation}
since $|G|\le cR_T$ at all price jumps.

Given $1$-admissible integrands $H^n$, first replace each by its corresponding
$G^n$, preserving its whole gain process. The reduction may depend on the
strategy; averaging the reduced portfolios still gives exactly the
corresponding average of the original gains. Lemma~\ref{average} yields finite
forward convex combinations $\widetilde H^n$ of the $G^n$ converging in
$L^2(\mu;\mathbb R^d)$.
Take a subsequence converging $\mu$-almost everywhere, and relabel it. Its
$n$th term is still a convex combination of $G^j$ with $j\ge n$, since
subsequence indices increase. Let $H$ be its
finite pointwise limit, set equal to zero where that limit fails to exist.
This defines a predictable process. If $A$ is the predictable exceptional set,
then
\[
 0=\mu(A)=\E\left[R_T^{-3}\int_0^T\one_Ac^{-2}\,dV\right]
 \quad\Longrightarrow\quad \int_0^T\one_A\,dV=0\quad\as
\]
Positivity of the weight therefore gives convergence at every price jump,
on a common full-probability set. Finitely many jumps on each path imply
\begin{equation}\label{uniform}
 \sup_{t\le T}|(\widetilde H^n\cdot S)_t-(H\cdot S)_t|
 \le\sum_{0<u\le T}|\widetilde H^n_u-H_u|\,|\Delta S_u|
 \longrightarrow0\quad\as
\end{equation}
Convex combinations preserve the lower bound $-1$, as does this uniform
limit. Thus $H\cdot S$ is $1$-admissible.

For $r>0$ set $C_r=C\cap\{f:\|f\|_\infty\le r\}$. We show directly that
$C_r$ is closed in $L^2(\Pp)$. Let $f_n\in C_r$ converge to $f$ in $L^2$,
and pass to an almost-surely convergent subsequence. Choose admissible gains
$X^n=H^n\cdot S$ with $X^n_T\ge f_n$. Since $f_n\ge-r$, the continuation
argument makes each $X^n$ $r$-admissible. Apply the preceding construction
to $H^n/r$ and rescale by $r$. It gives forward convex combinations of
$X^n$ converging uniformly almost surely to an $r$-admissible gain $X$.
The same convex combinations of $f_n$ converge almost surely to $f$, so
$X_T\ge f$. Also $|f|\le r$, hence $f\in C_r$.

Since $C_r$ is convex and norm closed in $L^2$, Hahn--Banach separation
makes it weakly closed there. The separating functionals
$f\mapsto\E[gf]$, $g\in L^2\subset L^1$, are also
$\sigma(L^\infty,L^1)$-continuous. Hence $C_r$ is weak-star closed in
$L^\infty$. The Krein--\v Smulian theorem \cite{KS} (see also \cite[Theorem~6]{Lurie}) states that a convex subset of a
dual Banach space is weak-star closed if its intersection with every closed
norm ball centered at the origin is weak-star closed. Applying it to $L^\infty=(L^1)^*$ proves
that $C$ is $\sigma(L^\infty,L^1)$-closed. NA therefore gives NFL,
and hence NFLVR $\Rightarrow$ NFL.

\smallskip
\noindent\emph{(iii) Construction of the martingale measure.}

Apply Theorem~\ref{KY} with $D=C$. We have already observed that $C$ is
a convex cone containing $-L^\infty_+$. NA gives
$C\cap L^\infty_+=\{0\}$, and part~(ii) proved the required weak-star
closedness. Set $d\Qq=Z\,d\Pp$; integrability,
strict positivity and normalization make $\Qq$ a countably additive
equivalent probability measure.
For a stopping time $\sigma$ such that $S^\sigma$ is bounded, $A\in\F_s$,
$s<t$, and a coordinate vector $e_i$, use the predictable strategy
\[
 H_u=e_i\one_A\one_{(s,t]}(u)\one_{[0,\sigma]}(u).
\]
Both $H\cdot S$ and $-H\cdot S$ are bounded, hence admissible. Their terminal
values belong to $C$, so separation in both signs gives
$\E_{\Qq}[\one_A(S^i_{t\wedge\sigma}-S^i_{s\wedge\sigma})]=0$.
These are exactly the martingale identities for $S^\sigma$. The localizing
sequence from Section~2 therefore makes $S$ a $\Qq$-local martingale.

\smallskip
\noindent\emph{(iv) An equivalent local martingale measure implies NA.}
Suppose $S$ is a $\Qq$-local martingale and $X=H\cdot S\ge-a$. Let $T_m$ be
the successive price-jump times, equal to $\infty$ after the last jump.
Set $\theta_m=\sigma_m\wedge T_m\wedge T$. Then $\theta_m=T$ eventually almost
surely, while $S^{\theta_m}$ is a bounded martingale with at most $m$ jumps
and total variation at most $2mb_m$.

Fix $m$, and abbreviate $\bar S=S^{\theta_m}$ and $\bar X=X^{\theta_m}=H\cdot\bar S$.
For each coordinate, the variation bound ensures that
$A\mapsto\E_{\Qq}\int_0^T\one_A\,d\bar S^i$ defines a finite signed
measure on the predictable sigma-field of $\Omega\times(0,T]$. It vanishes
on the generating rectangles
$A=B\times(s,t]$, $B\in\F_s$, by the martingale property. The monotone-class
theorem then shows $\E_{\Qq}(h\cdot\bar S)_T=0$ for every bounded predictable $h$.

For general $H$, put
\[
 \rho_n=\inf\{t\in(0,T]:\Delta\bar S_t\ne0,\ |H_t|>n\},\qquad
 Y^n=(H\one_{\{|H|\le n\}}\one_{[0,\rho_n]})\cdot\bar S,
\]
with $\inf\varnothing=\infty$. This is a stopping time: enumerate the jumps
of $\bar S$ and use that $H$ evaluated at a stopping time is measurable there.
The integrand defining $Y^n$ is bounded and predictable; in particular,
$\one_{[0,\rho_n]}$ is left-continuous and adapted. Before
$\rho_n$, $Y^n=\bar X$; at $\rho_n$ it omits that jump and stays at
$\bar X_{\rho_n-}\ge-a$. Thus $Y^n\ge-a$ and $\E_{\Qq}Y^n_T=0$.
There are finitely many price jumps, with finite values of $H$ at each, so
$Y^n=\bar X$ eventually on each path. Applying Fatou's lemma to the
nonnegative variables $a+Y^n_T$ gives
\[
 \E_{\Qq}[a+X_{\theta_m}]
 \le\liminf_{n\to\infty}\E_{\Qq}[a+Y^n_T]=a.
\]
In particular, $X_{\theta_m}$ is integrable and has nonpositive expectation.
Since $\theta_m=T$ eventually almost surely, a second application to
$a+X_{\theta_m}$ shows that $X_T$ is integrable and
$\E_{\Qq}X_T\le0$. If $X_T\ge0$, this inequality forces $X_T=0$
$\Qq$-almost surely, and hence $\Pp$-almost surely because $\Qq\sim\Pp$.
Thus NA holds. Combining this implication with parts~(i)--(iii) and
NFL $\Rightarrow$ NFLVR $\Rightarrow$ NA proves all the equivalences.
A bounded local martingale is a true martingale, giving the final assertion.
\end{proof}

\begin{remark}
Uniform separation of price vectors is sufficient but unnecessary: the
finite sets $E_m$ provide all the geometry used in the proof. For example,
$E=\{\sqrt n:n=0,1,2,\ldots\}^d$ is locally finite without a uniform separation.
The state-space condition concerns discounted prices.
\end{remark}

\begin{remark}
The local-martingale conclusion is necessary. Starting at one, let the price
double or jump to zero at deterministic times $t_n\uparrow T$, using independent
fair coin tosses and their natural filtration, and remain zero after its
first zero. Absorption occurs almost surely, so there are finitely many
jumps. Stopping upon first reaching $2^m$, or at $T$ if that level is never
reached, gives a bounded martingale and a localizing sequence.
Yet $S_T=0$ almost surely, so no equivalent measure makes $S$ a true
martingale; compare \cite{Ruf}. A pathwise finite jump count therefore differs
substantially from a deterministic finite number of periods.
\end{remark}

\appendix
\section{Functional-analytic tools}\label{app:tools}
\setcounter{theorem}{0}
\renewcommand{\thetheorem}{\thesection.\arabic{theorem}}
\renewcommand{\theHtheorem}{appendix.\thesection.\arabic{theorem}}

\subsection{Convex averaging in a Hilbert space}

\begin{lemma}[$L^2$ Koml\'os lemma, {\cite[Section~2.1]{BSV}}]\label{average}
For every bounded sequence $(u_n)$ in a Hilbert space there are
$h_n\in\conv\{u_j:j\ge n\}$ converging in norm. Each convex combination
uses finitely many terms with deterministic nonnegative coefficients summing
to one.
\end{lemma}

\begin{proof}
Put $D_n=\conv\{u_j:j\ge n\}$ and $m_n=\inf_{h\in D_n}\|h\|^2$.
Then $m_n\uparrow m<\infty$. Choose $h_n\in D_n$ with
$\|h_n\|^2\le m_n+1/n$. For $i,j\ge n$, the midpoint of $h_i,h_j$ lies
in $D_n$. The parallelogram identity gives
\[
 \|h_i-h_j\|^2
 =2\|h_i\|^2+2\|h_j\|^2-4\|(h_i+h_j)/2\|^2
 \le4(m-m_n)+2/i+2/j.
\]
Hence $(h_n)$ is Cauchy, and completeness gives its limit.
\end{proof}

\subsection{A strictly positive separating density}

We use the following form of the Kreps--Yan theorem
\cite{Kreps,Yan}; see also \cite[Appendix, Lemma F]{K}.

\begin{theorem}[Kreps--Yan]\label{KY}
Let $D\subset L^\infty(\Pp)$ be a convex $\sigma(L^\infty,L^1)$-closed cone
such that $-L^\infty_+\subset D$ and $D\cap L^\infty_+=\{0\}$.
Then there exists $Z\in L^1(\Pp)$ satisfying
\[
 Z>0\quad\as,\qquad \E Z=1,\qquad \E[Zf]\le0\quad\text{for every }f\in D.
\]
\end{theorem}

\begin{proof}
For every event $A$ of positive probability, $\one_A\notin D$.
Hahn--Banach separation in the stated weak-star topology gives $z_A\in L^1$
with $\E[z_A\one_A]>0$ and $\E[z_Af]\le0$ for $f\in D$. The cone property
makes the separating bound zero; $-L^\infty_+\subset D$ forces $z_A\ge0$.
Normalize $\E z_A=1$.

Let $\mathcal Z$ be all nonnegative normalized separating densities and put
$\alpha=\sup_{z\in\mathcal Z}\Pp(z>0)$. Choose $z_n\in\mathcal Z$ with
$\Pp(z_n>0)\to\alpha$. The $L^1$-convergent sum
$Z=\sum_{n\ge1}2^{-n}z_n$ belongs to $\mathcal Z$. Its positive set contains
each $\{z_n>0\}$, so $\Pp(Z>0)=\alpha$. If $A=\{Z=0\}$ had positive probability, the density
$(Z+z_A)/2$ would belong to $\mathcal Z$ and be positive on a set of
probability greater than $\alpha$, a contradiction.
\end{proof}

\section*{Acknowledgments}
This work was supported by JSPS KAKENHI Grant Number 25K07112.

\section*{Declaration of AI use}
OpenAI Codex was used to assist with mathematical discussion and proof
review, literature searches, exposition, and LaTeX editing. The author
retains responsibility for the mathematical content and references.

\end{document}